\documentclass[journal]{IEEEtran}
\usepackage[T1]{fontenc}% optional T1 font encoding

\usepackage{amssymb}
\usepackage{amsmath}
\usepackage{cite}
\usepackage{url}
\usepackage{cite,graphicx,amsmath,amssymb}
\usepackage{fancyhdr}
\usepackage{mdwmath}
\usepackage{mdwtab}
\usepackage{caption}
\usepackage{amsthm}
\usepackage{setspace}
\usepackage{hyperref}
\usepackage{algorithm}
\usepackage{algorithmic}
\usepackage{multirow}
\usepackage{makecell}
\usepackage{mathtools}
\usepackage{subcaption}
\usepackage{bm}
\usepackage{tikz}
\usepackage{xurl}
\usepackage{stfloats}
\usepackage{mathrsfs}

\usepackage{booktabs}
\usepackage{multirow}
\usepackage{siunitx}
\usepackage{balance}
\usepackage{bm}

\usepackage{array}
\usepackage{multirow,tabularx}

\hypersetup{colorlinks=true,
	linkcolor=blue,
	citecolor=blue,      
	urlcolor=black,
}

\theoremstyle{definition}

\newtheorem{theorem}{Theorem}

\newtheorem{remark}{Remark}

\graphicspath{{./figure/}}

\begin{document}
	
	\title{C-PASS: Center-Fed Pinching Antenna System}

	\author{Xu Gan,~\IEEEmembership{Member, IEEE}, and Yuanwei Liu,~\IEEEmembership{Fellow,~IEEE} 		\vspace{-3mm}
		\thanks{(Corresponding author: Yuanwei Liu.)}
		\thanks{Xu Gan is with the Department of Electrical and Computer Engineering, The University of Hong Kong, Hong Kong (e-mail: eee.ganxu@hku.hk).}
		\thanks{Yuanwei Liu is with the Department of Electrical and Computer Engineering, The University of Hong Kong, Hong Kong, and also with the Department of Electronic Engineering, Kyung Hee University, Yongin-si, Gyeonggi-do 17104, South Korea (e-mail: yuanwei@hku.hk).}
	}

	\maketitle
	
\begin{abstract}
	A novel architecture of the center-fed pinching antenna system (C-PASS) is proposed. In contrast to the conventional end-fed PASS, signals are fed from the center input ports and propagate towards both sides of the waveguide. By doing so, spatial-multiplexing gain can be achieved in a single waveguide. Based on the proposed C-PASS, closed-form expressions for the degree of freedom (DoF) and power scaling laws are derived. These theoretical results reveal that C-PASS can achieve \emph{twice} the DoF and an additional multiplexing gain of $\mathcal{O}(P_T \ln^4 N/N^2)$ compared to the conventional PASS, where $P_T$ and $2N$ represent the transmit power and the total pinching antenna number, respectively. Numerical results are provided to demonstrate that substantial capacity improvements can be achieved through the enhanced DoF and multiplexing gain of the C-PASS.		
		

\end{abstract}

\section{Introduction}
\IEEEPARstart{R}{ecently}, pinching antenna systems (PASSs) have attracted increasing attention as a promising architecture for replacing part of high-loss wireless propagation with stable, low-loss wired transmission~\cite{liu2025pinching,liu2025pinching2,yang2025pinching}. A PASS consists of dielectric waveguides and separate pinching antennas (PAs), which radiate guided signals into free space through electromagnetic coupling~\cite{wang_model}. By adjusting the PA deployment positions, PASSs can flexibly control the wired and wireless propagation distances, enabling applications such as physical layer security~\cite{shan2025secure}, energy efficiency enhancement~\cite{gan2025dual}, and unmanned aerial vehicle communications~\cite{lv2025pinching}.

Despite these advantages, existing studies have mainly considered end-fed PASSs~\cite{liu2025pinching,liu2025pinching2,yang2025pinching,shan2025secure,gan2025dual,lv2025pinching,ouyang2025array,wang_model}, where a terminal source excites all PAs through a single serial waveguide path. This ``Single-Input Multiple-Radiation'' topology yields a rank-one effective channel, limiting the degrees of freedom (DoF) to $1$ and forcing multi-user transmission or channel estimation to rely on time/frequency division. To overcome this rank deficiency, we propose a center-fed PASS (C-PASS) architecture that can \emph{double} the DoF of conventional PASS. C-PASS preserves the established PA radiation model and only changes the feeding framework from unidirectional serial excitation to bidirectional center-fed excitation. It can be implemented using T-junction waveguides~\cite{hong2018high,reichel2016broadband,li2015broadband}, where a tunable power splitter launches controllable bidirectional propagation from the waveguide center. However, the communication model and performance limits of this center-fed architecture remain unexplored.

The main contributions are summarized as follows: i) We propose the C-PASS architecture, which enables controllable bidirectional signal transmission through tunable power splitters. ii) We derive the DoF of C-PASS and conventional end-fed PASS in \emph{Theorem 1}, showing that C-PASS achieves twice the DoF of conventional PASS. The doubled DoF enables spatial multiplexing in a single waveguide, addressing the rank-deficiency bottleneck of conventional PASS. iii) We analyze the power scaling laws for both architectures in \emph{Theorem 2}, showing that while both achieve an array gain of $\mathcal{O}(\ln^2 N/N)$, C-PASS additionally offers a multiplexing gain of $\mathcal{O}(P_T \ln^4 N/N^2)$, where $2N$ is the number of all PAs and $P_T$ is the transmit power. iv) Numerical results validate the theoretical analyses in \emph{Theorems 1} and \emph{2} and demonstrate the superiority of C-PASS over conventional PASS.

\section{Proposed C-PASS Architecture and\\ Signal Model}\label{sec2}
Fig.~\ref{fig:concept_PASS} illustrates the proposed C-PASS architecture, where a tunable waveguide T-junction feeds the forward and backward PA groups. The input branch is denoted as Port~1, while the two output branches toward the forward and backward directions are denoted as Port~2 and Port~3, respectively. From microwave network theory~\cite{pozar2011microwave}, the T-junction can be modeled as a reciprocal three-port network with scattering matrix $\mathbf{S}\in\mathbb{C}^{3\times3}$, where $|[\mathbf{S}]_{2,1}|^2$ and $|[\mathbf{S}]_{3,1}|^2$ characterize the power delivered from Port~1 to the two output branches. Unlike static T-junctions whose splitting ratios are fixed by geometry, tunable implementations can adjust these transmission coefficients by perturbing the electromagnetic boundary conditions, e.g., via gyromagnetic materials under external magnetic fields~\cite{hong2018high} or mechanically movable septa~\cite{reichel2016broadband}. Mathematically, let $\mathbf{x}_{\text{in}}$ denote the incident signal at the input port. Then, the forward-propagating signal $\mathbf{x}_{\text{in}}^{\text{F}}$ and backward-propagating signal $\mathbf{x}_{\text{in}}^{\text{B}}$ are modeled as:
\begin{equation}
	\mathbf{x}_{\text{in}}^{\text{F}} = \sqrt{\beta_{\text{F}}} \mathbf{x}_{\text{in}}, \quad \mathbf{x}_{\text{in}}^{\text{B}} = \sqrt{\beta_{\text{B}}} \mathbf{x}_{\text{in}},
\end{equation}
where the power splitting ratio $\beta_{\chi}$ for $\chi \in \{\text{F}, \text{B}\}$ is defined by $\beta_{\text{F}} = |[\mathbf{S}]_{2,1}|^2$ and $\beta_{\text{B}} = |[\mathbf{S}]_{3,1}|^2$. Due to passivity of the T-junction,
\begin{equation}\label{sec2_beta}
	\beta_{\text{F}} + \beta_{\text{B}} \le 1.
\end{equation}
After power splitting, $\mathbf{x}_{\text{in}}^{\text{F}}$ and $\mathbf{x}_{\text{in}}^{\text{B}}$ propagate outward from the center input port along the forward and backward waveguide directions, respectively. When the guided waves encounter the PAs, part of the energy is radiated into free space through electromagnetic coupling. For the $n$-th PA in the $\chi$-direction, let $\mathbf{x}_{n}^{\chi, \text{inc}}$, $\mathbf{x}_{n}^{\chi, \text{rad}}$, and $\mathbf{x}_{n}^{\chi, \text{thr}}$ denote the incident, radiated, and through signals, respectively. Assuming negligible radiation loss and phase discontinuity, these signals are modeled as:
\begin{align}
	 \mathbf{x}_{n}^{\chi, \text{rad}} = \sqrt{\delta_n^{\chi}} \mathbf{x}_{n}^{\chi, \text{inc}}, \
	 \mathbf{x}_{n}^{\chi, \text{thr}} = \sqrt{1-\delta_n^{\chi}} \mathbf{x}_{n}^{\chi, \text{inc}},
\end{align}
where $\delta_n^{\chi} \in [0,1]$ denotes the radiation power ratio of the $n$-th PA, which can be adjusted via the coupling length between the PA and the waveguide structure~\cite{wang_model}. Then, the through signal $\mathbf{x}_{n}^{\chi, \text{thr}}$ continues its propagation along the waveguide.  After traveling a distance of $d_n^{\chi}$, the signal arrives at the $(n+1)$-th PA and serves as the incident signal, formulated as
\begin{equation}
	\mathbf{x}_{n+1}^{\chi, \text{inc}} = \exp\left( -j k_g d_n^{\chi} \right) \mathbf{x}_{n}^{\chi, \text{thr}}, 
\end{equation}
where $k_g=\frac{2\pi}{\lambda_g}$ is the propagation wavenumber in the waveguide, and $\lambda_g$ is the effective wavelength in the waveguide medium. The in-waveguide wavelength $\lambda_g$ is related to the free-space wavelength $\lambda_0$ by the effective refractive index $n_{\text{eff}}$ of the waveguide as $\lambda_g=\frac{\lambda_0}{n_{\text{eff}}}$. Based on the derived mathematical expressions, we can formulate a closed-form expression for the signal radiated by the $n$-th PA in the $\chi$-propagation direction as
\begin{equation}\label{sec2_PA}
	\mathbf{x}_{n}^{\chi, \text{rad}} = \sqrt{\beta_{\chi} \xi_n^{\chi}} \exp\left( -j k_g D_n^{\chi} \right) \mathbf{x}_{\text{in}},
\end{equation}
where the term $\xi_n^{\chi} = \delta_n^{\chi} \prod_{m=1}^{n-1} (1 - \delta_m^{\chi})$ represents the cumulative radiation coefficient for the $n$-th PA, and $D_n^{\chi}=\sum_{m=0}^{n-1} d_m^{\chi}$ denotes the total propagation distance from the input port to the $n$-th PA, in the $\chi$-direction.
For a conventional end-fed PASS, the radiated signal of the $n$-th PA is given by $\mathbf{x}_{n}^{\text{E,rad}}=\sqrt{\xi_n}\exp(-j k_g D_n)\mathbf{x}_{\text{in}}$, which is obtained by setting $\beta_{\text{F}}=1$, $\beta_{\text{B}}=0$ or $\beta_{\text{F}}=0$, $\beta_{\text{B}}=1$ in \eqref{sec2_PA}, where $\xi_n=\delta_n\prod_{m=1}^{n-1}(1-\delta_m)$ and $D_n$ denotes the accumulated distance from the input port to the $n$-th PA.

\begin{figure}[t]
	\centering
	\begin{minipage}[c]{0.48\linewidth}
		\centering
		\includegraphics[width=0.94\linewidth]{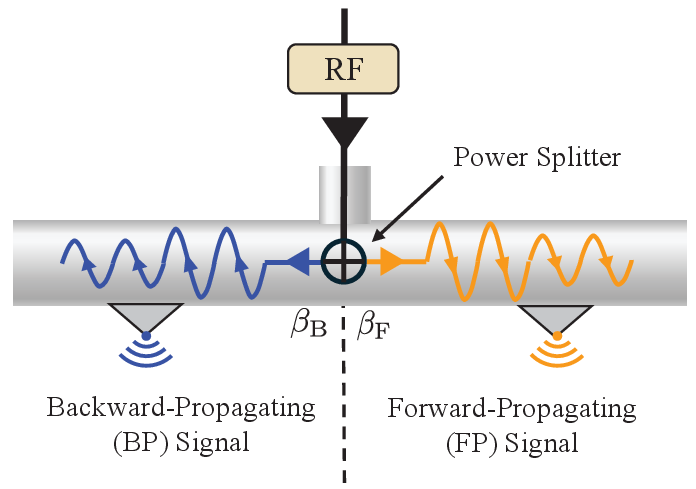}
	\end{minipage}
	\hfill
	\begin{minipage}[c]{0.48\linewidth}
		\centering
		\includegraphics[width=\linewidth]{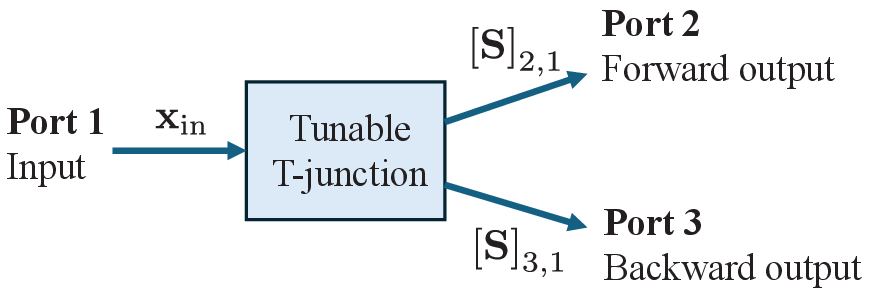}
	\end{minipage}
	\caption{The architecture of center-fed PASS.}
	\label{fig:concept_PASS}
	\vspace{-5mm}
\end{figure}

\begin{remark}
\normalfont The term ``center-fed'' refers to the feeding architecture rather than the exact geometric center of the waveguide. An input port is center-fed when it is placed between two PA groups and excites guided waves toward both sides. In contrast, conventional PASS architectures in \cite{liu2025pinching,liu2025pinching2,yang2025pinching,shan2025secure,gan2025dual,lv2025pinching,ouyang2025array,wang_model} are end-fed, where the input port is located at a waveguide terminal and the guided signal traverses one cascaded PA sequence.
\end{remark}

\section{DoF and Power Scaling Law Analysis}
This section compares the DoF and power scaling laws of C-PASS and conventional end-fed PASS. For a fair comparison, both architectures use two input ports to serve two users with the same total number of $2N$ PAs, as shown in Fig.~\ref{fig:system_comparison}. The configurations are defined as follows:
\begin{itemize}
	\item \textbf{Center-Fed PASS:} The two signals are divided via two power splitters into the forward- and backward-direction of the waveguide. Consequently, these forward- and backward-direction signals are radiated by the $N$ forward-direction PAs (FPAs) and $N$ backward-direction PAs (BPAs), respectively.
	
	\item \textbf{Conventional End-Fed PASS:} The two signals propagate along the same transmission direction from the terminal input ports. Consequently, these signals are both radiated by the cascaded entire $2N$ PAs.
\end{itemize}

\begin{figure}[t!]
	\centering
	\subfloat[C-PASS.]{
		\includegraphics[width=0.8\linewidth]{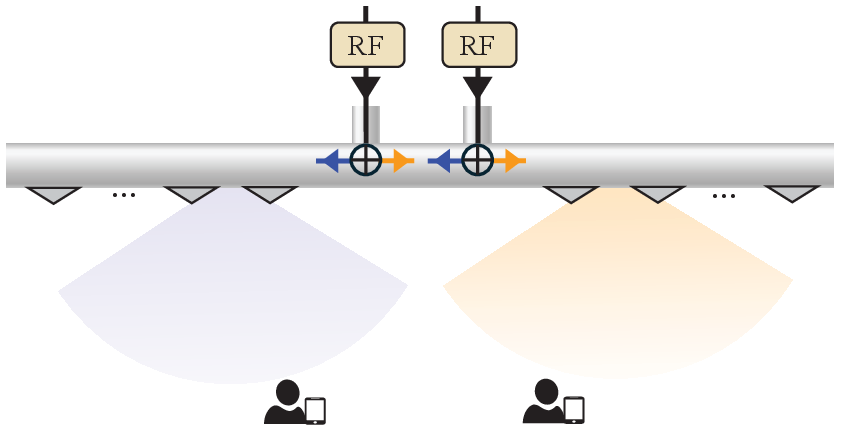}
		\label{fig:system_center}
		\vspace{-1mm}
	}
	\hfill
	\subfloat[Conventional end-fed PASS.]{
		\includegraphics[width=0.8\linewidth]{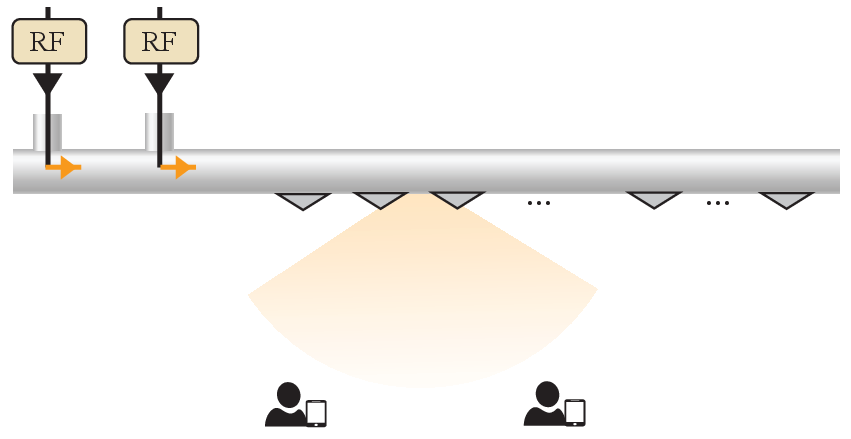}
		\label{fig:system_end}
		\vspace{-1mm}
	}
	\caption{Illustration of C-PASS and conventional end-fed PASS.}
	\label{fig:system_comparison}
	\vspace{-2mm}
\end{figure}

\subsection{Channel Models}
For clarity, the input port and user close to the FPA are denoted as the forward-direction input port (FIN) and forward-direction user (FU), while those close to the BPA are denoted as the backward-direction input port (BIN) and backward-direction user (BU). Let $\beta_{\chi_1 \chi_2}$ denote the power splitting ratio from $\chi_1$-IN to the $\chi_2$-direction, where $\chi_1, \chi_2 \in \{\text{F}, \text{B}\}$. From~\eqref{sec2_beta}, $\beta_{\text{FF}} + \beta_{\text{FB}} \le 1$ and $\beta_{\text{BF}} + \beta_{\text{BB}} \le 1$. The effective in-waveguide channel is then $\mathbf{G}_{\text{C}} =  \Big[\begin{smallmatrix} \sqrt{\beta_{\text{FF}}} \left( \mathbf{g}_{\text{C}}^{\text{FF}} \right)^T & \sqrt{\beta_{\text{FB}}} \left( \mathbf{g}_{\text{C}}^{\text{FB}} \right)^T  \\ \sqrt{\beta_{\text{BF}}} \left( \mathbf{g}_{\text{C}}^{\text{BF}} \right)^T & \sqrt{\beta_{\text{BB}}} \left( \mathbf{g}_{\text{C}}^{\text{BB}} \right)^T  \end{smallmatrix} \Big]$, where $\mathbf{g}_{\text{C}}^{\chi_1 \chi_2}$ is the in-waveguide propagation vector from $\chi_1$-IN to the $\chi_2$-PA. Based on~\eqref{sec2_PA}, its $n$-th element is
\begin{subequations}\label{center_channel_g}
	\begin{align}
		& \left[ \mathbf{g}_{\text{C}}^{\chi \chi} \right]_n = \exp\left( - j k_g n L_{\text{pa}}\right),\\
		&  \left[ \mathbf{g}_{\text{C}}^{\chi \bar{\chi}} \right]_n = \exp\left( - j k_g \left( L_{\text{in}} + n L_{\text{pa}} \right) \right),
	\end{align}
\end{subequations}
where $\bar{\chi}$ represents the complement of $\chi$ in the set $\{\text{F}, \text{B} \}$. Here, $L_{\text{pa}}$ represents the inter-element spacing of the PAs, while $L_{\text{in}}$ denotes the physical separation between the two input ports. Then, we introduce the PA radiation matrix $\mathbf{\Sigma}_{\text{C}} = \text{blkdiag}(\bm{\Sigma}_{\text{C}}^{\text{F}}, \bm{\Sigma}_{\text{C}}^{\text{B}})$, where $\bm{\Sigma}_{\text{C}}^{\text{F}}$ and $\bm{\Sigma}_{\text{C}}^{\text{B}}$ represent the diagonal radiation matrix for FPAs and BPAs, respectively, and their $n$-th diagonal element is $\sqrt{\xi_n^{\chi}}$. For analytical tractability, we adopt a uniform radiation scheme where the radiation coefficient for all PAs is $\xi_n^{\chi} = 1/N$, $\forall n, \chi$. Subsequent to the radiation process, the signals undergo wireless propagation to serve the communication users. Let $\mathbf{H}_{\text{C}} = \big[\begin{smallmatrix}
\mathbf{h}_{\text{C}}^{\text{FF}} &\mathbf{h}_{\text{C}}^{\text{FB}} \\
\mathbf{h}_{\text{C}}^{\text{BF}} & \mathbf{h}_{\text{C}}^{\text{BB}} \end{smallmatrix} \big]$ denote the aggregate channel from the $2N$ PAs to the two users, where $\mathbf{h}_{\text{C}}^{\chi_1 \chi_2}$ represents the channel vector from the $\chi_1$-PA to the $\chi_2$-user. Assuming line-of-sight (LoS) free-space propagation, the $n$-th element of $\mathbf{h}_{\text{C}}^{\chi_1 \chi_2}$ is formulated as
\begin{equation}\label{center_channel_h}
	[\mathbf{h}_{\text{C}}^{\chi_1 \chi_2}]_n = \eta \frac{\exp(-j k_0 r_n^{\chi_1 \chi_2})}{r_n^{\chi_1 \chi_2}},
\end{equation}
where $\eta$ and $k_0=\frac{2\pi}{\lambda_0}$ are the path-loss coefficient and wavenumber of the free-space propagation, respectively. The term $r_n^{\chi_1 \chi_2}$ denotes the Euclidean distance between the $n$-th $\chi_1$-PA and the $\chi_2$-user, with $r_n^{\chi \chi}=\sqrt{Y_{\chi}^2+n^2L_{\text{pa}}^2}$ and $r_n^{\chi \bar{\chi}}=\sqrt{Y_{\bar{\chi}}^2+(L_{\text{in}}+nL_{\text{pa}})^2}$. The two users are assumed to be located at the horizontal positions of the FIN and BIN, and $Y_{\text{F}}$ and $Y_{\text{B}}$ denote their vertical distances from the waveguide.

To model the end-to-end effective channel from the input ports to the users in the C-PASS, we cascade the derived channel models~\eqref{center_channel_g} and \eqref{center_channel_h}. The overall effective channel is given by $\mathbf{H}_{\text{C}}^{\text{eff}} = \mathbf{G}_{\text{C}} \bm{\Sigma}_{\text{C}} \mathbf{H}_{\text{C}}$, where the entry $[\mathbf{H}_{\text{C}}^{\text{eff}}]_{\chi_1 \chi_2}$ denotes the complex channel coefficient from the $\chi_1$-IN to the $\chi$-user, is explicitly formulated as
\begin{subequations}\label{center_channel_effective}
	\begin{align}
		\left[\mathbf{H}_{\text{C}}^{\text{eff}}\right]_{\text{FF}} &= \sqrt{\beta_{\text{FF}}} A_{\text{C}}^{\text{FF}} + \sqrt{\beta_{\text{FB}}} A_{\text{C}}^{\text{BF}} \exp\left( -j k_g  L_{\text{in}} \right), \\
		\left[\mathbf{H}_{\text{C}}^{\text{eff}}\right]_{\text{FB}} &= \sqrt{\beta_{\text{FF}}} A_{\text{C}}^{\text{FB}} + \sqrt{\beta_{\text{FB}}} A_{\text{C}}^{\text{BB}} \exp\left( -j k_g  L_{\text{in}} \right) ,\\
		\left[\mathbf{H}_{\text{C}}^{\text{eff}}\right]_{\text{BF}}  &= \sqrt{\beta_{\text{BF}}} A_{\text{C}}^{\text{FF}} \exp\left( -j k_g  L_{\text{in}} \right) + \sqrt{\beta_{\text{BB}}} A_{\text{C}}^{\text{BF}} ,  \\
		\left[\mathbf{H}_{\text{C}}^{\text{eff}}\right]_{\text{BB}}  &= \sqrt{\beta_{\text{BF}}} A_{\text{C}}^{\text{FB}} \exp\left( -j k_g  L_{\text{in}} \right) + \sqrt{\beta_{\text{BB}}} A_{\text{C}}^{\text{BB}} ,
	\end{align}
\end{subequations}
where
\begin{equation}\label{center_A}
	A_{\text{C}}^{\chi_1 \chi_2} = \frac{\eta}{\sqrt{N}} \sum_{n=1}^N \exp(-j k_g n L_{\text{pa}}) \frac{\exp\left(-j k_0 r_n^{\chi_1 \chi_2}\right)}{r_n^{\chi_1 \chi_2}}.
\end{equation}

For the end-fed PASS channel model, the formulation follows the same structure as the C-PASS channel model, and is thus not repeated in detail here. The primary differences lie in the power splitting coefficients and PA radiation matrices. Due to the unidirectional signal propagation along the waveguide, no power splitting occurs, i.e., $\beta_{\text{FF}}=\beta_{\text{BF}}=1$ or $\beta_{\text{FB}}=\beta_{\text{BB}}=1$ for the channel $\mathbf{G}_{\text{E}}$. Furthermore, since the input signal propagates through all $2N$ PAs, the radiation matrix is given by $\mathbf{\Sigma}_{\text{E}} = \frac{1}{\sqrt{2N}} \mathbf{I}_{2N}$.

\subsection{DoF Analysis}
Assuming equal power allocation across the two input ports, the ergodic capacity of the C-PASS and the end-fed PASS is $C_{\varpi} = \log_2 \det \left( \mathbf{I}_2 + \frac{P_T}{2 N_0} \mathbf{H}_{\varpi}^{\text{eff}} \left(\mathbf{H}_{\varpi}^{\text{eff}}\right)^H  \right)$, where $\varpi \in \{ \text{C}, \text{E} \}$ represents the center-fed and end-fed architectures, $P_T$ is the transmit power, and $N_0$ denotes the additive noise power. The DoF is defined as the high-SNR pre-logarithmic factor, i.e., $\text{DoF}_{\varpi} = \lim_{P_T \to \infty} C_{\varpi}/\log_2(P_T/N_0)$. The following theorem establishes the DoF achievable by the center-fed and end-fed PASS architectures.

\begin{theorem}
	\emph{For a PASS configured with two input ports, the achievable DoF for the C-PASS and the conventional end-fed PASS are $\text{DoF}_{\text{C}} = 2$ and $\text{DoF}_{\text{E}} = 1$, respectively.}
\end{theorem}
\begin{proof}
	Direct evaluation of this limit is analytically intractable. To address this, we exploit the fundamental equivalence between the DoF and the effective channel matrix rank, denoted as $R_{\varpi} = \text{rank}(\mathbf{H}_{\varpi}^{\text{eff}})$. This equivalence can be justified through the singular value decomposition of $\mathbf{H}_{\varpi}^{\text{eff}}$ with non-zero singular values $\{\sigma_i\}_{i=1}^{R_{\varpi}}$. Then, the capacity in asymptotic regime can be expanded as $C_{\varpi} =  \sum_{i=1}^{R_{\varpi}} \log_2 \left( 1 + P_T/(2 N_0) \sigma_i^2 \right) \overset{P_T \to \infty}{=} R_{\varpi} \log_2 \left( P_T/N_0 \right) + O(1)$. Substituting this expansion into the definition of DoF yields $\text{DoF}_{\varpi} = R_{\varpi}$. Then, we analyze the determinant value of the two architectures to derive $R_{\varpi}$. By substituting the channel coefficients~\eqref{center_channel_effective}, the determinant value of $\mathbf{H}_{\text{C}}^{\text{eff}}$ is
		\begin{equation}\label{center_det}
			\begin{aligned}
				\det \left( \mathbf{H}_{\text{C}}^{\text{eff}} \right) &= \left( A_{\text{C}}^{\text{FF}} A_{\text{C}}^{\text{BB}}\! - \! A_{\text{C}}^{\text{FB}} A_{\text{C}}^{\text{BF}} \right) \\
				&\times  \left[\sqrt{\beta_{\text{FF}} \beta_{\text{BB}}}\! -\! \sqrt{\beta_{\text{FB}} \beta_{\text{BF}}} \exp\left( -2 j k_g  L_{\text{in}}\right) \right].
			\end{aligned}
		\end{equation}
		The determinant in \eqref{center_det} consists of a geometry factor and a feeding factor. The geometry factor $A_{\text{C}}^{\text{FF}} A_{\text{C}}^{\text{BB}}-A_{\text{C}}^{\text{FB}} A_{\text{C}}^{\text{BF}}$ becomes zero only when the two aggregate PA-user response vectors are exactly proportional, which requires a degenerate user-PA geometry with identical distance-dependent amplitude and phase ratios for the two PA groups. The feeding factor $\sqrt{\beta_{\text{FF}}\beta_{\text{BB}}}-\sqrt{\beta_{\text{FB}}\beta_{\text{BF}}}\exp(-2jk_gL_{\text{in}})$ becomes zero only when $\beta_{\text{FF}}\beta_{\text{BB}}=\beta_{\text{FB}}\beta_{\text{BF}}$ and $L_{\text{in}}=q\lambda_g/2$, $q\in\mathbb{Z}$, hold simultaneously. Hence, under non-degenerate user-PA geometry and feeding configurations, $\det(\mathbf{H}_{\text{C}}^{\text{eff}})\neq0$, so $R_{\text{C}}=2$ and $\mathrm{DoF}_{\text{C}}=2$. In contrast, the end-fed PASS always radiates through one serial PA response, yielding a rank-one effective channel with $R_{\text{E}}=1$ and $\mathrm{DoF}_{\text{E}}=1$. This rank-one property also holds for an end-fed PASS with multiple input ports, since all input signals still share the same serial PA response.
\end{proof}

\begin{remark}
\normalfont The DoF result above is not restricted to the symmetric LoS deployment. Its key requirement is a non-zero C-PASS determinant, rather than symmetric user locations. For asymmetric deployments, the PA-user distances and aggregate responses are replaced by their geometry-dependent counterparts, while the determinant keeps the same factorized structure. Thus, under non-degenerate geometry and feeding configurations, the effective channel remains full-rank. For non-line-of-sight (NLoS)/Rician propagation, the channel can be viewed as the LoS response plus a zero-mean scattered component, so exact rank loss occurs only on a measure-zero set of channel realizations. Hence, asymmetric or NLoS/Rician channels may change finite-SNR gains, but they do not remove the structural DoF advantage of C-PASS.
\end{remark}

\subsection{Power Scaling Law Analysis}
With the obtained DoF characterizing the high-SNR capacity slope of the C-PASS, we next derive specific power scaling laws to explicitly evaluate the achievable capacity improvement. To facilitate the derivation of power scaling law, the capacity expression is reformulated as 
\begin{equation}
	C_{\varpi} \!=\! \log_2\! \left(\! 1 \!+\! \frac{P_T}{2 N_0} \left\| \mathbf{H}_{\varpi}^{\text{eff}} \right\|_F^2\! +\! \left( \frac{P_T}{2 N_0} \right)^2 \! \left| \det(\mathbf{H}_{\varpi}^{\text{eff}}) \right|^2\! \right)\!\!,
\end{equation}
by invoking the equation $\det(\mathbf{I}_2 + \mathbf{X}) = 1 + \text{tr}(\mathbf{X}) + \det(\mathbf{X})$ for any matrix $\mathbf{X} \in \mathbb{C}^{2 \times 2}$. This expansion allows us to decompose the effective channel gain $G_{\varpi}$ into
\begin{equation}\vspace{-1mm}
	G_{\varpi} =  \underbrace{\left\| \mathbf{H}_{\varpi}^{\text{eff}} \right\|_F^2}_{G_{\varpi}^{\text{A}}: \ \text{Array Gain}} + \underbrace{\frac{P_T}{2 N_0} \left| \det(\mathbf{H}_{\varpi}^{\text{eff}}) \right|^2}_{G_{\varpi}^{\text{M}}: \ \text{Multiplexing Gain}}.
\end{equation} 

We proceed to analyze the power scaling laws for the C-PASS. For tractability, we assume the power splitting ratios are equal, i.e., $\beta_{\chi_1 \chi_2}=\frac{1}{2}$. By incorporating this setting into the formulations derived in~\eqref{center_channel_effective} and \eqref{center_det}, the array gain and multiplexing gain are expressed as $G_{\text{C}}^{\text{A}}=|A_{\text{C}}^{\text{FF}}|^2+|A_{\text{C}}^{\text{FB}}|^2+|A_{\text{C}}^{\text{BF}}|^2+|A_{\text{C}}^{\text{BB}}|^2+2\Re\{A_{\text{C}}^{\text{FF}}(A_{\text{C}}^{\text{BF}})^*+A_{\text{C}}^{\text{FB}}(A_{\text{C}}^{\text{BB}})^*\}\cos(k_gL_{\text{in}})$ and $G_{\text{C}}^{\text{M}}=\frac{P_T}{4N_0}(1-\cos(2k_gL_{\text{in}}))|A_{\text{C}}^{\text{FF}}A_{\text{C}}^{\text{BB}}-A_{\text{C}}^{\text{FB}}A_{\text{C}}^{\text{BF}}|^2$, respectively.
Both gains depend on the input-port separation $L_{\text{in}}$ through the phase terms. Choosing $L_{\text{in}} = \frac{\lambda_g}{4} (1+2k), k \in \mathbb{Z}$, maximizes the multiplexing capability while maintaining a robust array gain. Under this setting, the gains simplify to
\begin{subequations}
	\begin{align}
		& G_{\text{C}}^{\text{A}} = |A_{\text{C}}^{\text{FF}}|^2 + |A_{\text{C}}^{\text{FB}}|^2 + |A_{\text{C}}^{\text{BF}}|^2 + |A_{\text{C}}^{\text{BB}}|^2, \\
		& G_{\text{C}}^{\text{M}} = \frac{P_T}{2N_0} \left| A_{\text{C}}^{\text{FF}} A_{\text{C}}^{\text{BB}} - A_{\text{C}}^{\text{FB}} A_{\text{C}}^{\text{BF}} \right|^2.
	\end{align} 
\end{subequations}
The expression of $G_{\text{C}}^{\text{M}}$ shows that the multiplexing term originates from the non-zero determinant enabled by the C-PASS architecture. To obtain a tractable scaling law, we use the fine-tuning position strategy in~\cite{gan2025dual}. Specifically, the FPAs are phase-aligned at the FU and the BPAs at the BU. Under this phase alignment, these direct-link channel gains become dominant:
\begin{equation}\label{hat_A_chi}
	|\hat{A}_{\text{C}}^{\chi \chi}| = \frac{\eta}{\sqrt{N}} \sum_{n=1}^N \frac{1}{\sqrt{Y_{\chi}^2 + n^2 L_{\text{pa}}^2 }}.
\end{equation}
In contrast, the cross-link gain $|\hat{A}_{\text{C}}^{\chi \bar{\chi}}|$ combines $N$ paths with random phases, rendering it negligible compared to the dominant term $|\hat{A}_{\text{C}}^{\chi \chi}|$. The distinct dominance of $|\hat{A}_{\text{C}}^{\chi \chi}|$ over $|\hat{A}_{\text{C}}^{\chi \bar{\chi}}|$, i.e., $|\hat{A}_{\text{C}}^{\chi \chi}| \gg |\hat{A}_{\text{C}}^{\chi \bar{\chi}}|$, becomes increasingly pronounced as $N\! \to\! \infty$. Leveraging this asymptotic dominance, we focus on these dominant terms to establish the power scaling laws in the following theorem.
\begin{theorem}
	\emph{Under the configurations $L_{\text{in}} = \frac{\lambda_g}{4} (1+2k), k \in \mathbb{Z}$ and the asymptotic regime of large $N$, the C-PASS achieves the array gain and multiplexing gain scaling as $\mathcal{O}\left( \frac{\ln^2 N}{N} \right)$ and $\mathcal{O}\left( P_T \frac{\ln^4 N}{N^2} \right)$, respectively. In contrast, the end-fed PASS exhibits an array gain scaling of $\mathcal{O}\left( \frac{\ln^2 N}{N} \right)$, while yielding zero multiplexing gain.}
\end{theorem}
\begin{proof}
	Under the specific input port spacing and phase-aligned configuration in the asymptotic regime $N \to \infty$, the dominant array and multiplexing gains are given by $\bar{G}_{\text{C}}^{\text{A}}=|\hat{A}_{\text{C}}^{\text{FF}}|^2+|\hat{A}_{\text{C}}^{\text{BB}}|^2$ and $\bar{G}_{\text{C}}^{\text{M}}=\frac{P_T}{2N_0}|\hat{A}_{\text{C}}^{\text{FF}}\hat{A}_{\text{C}}^{\text{BB}}|^2$.
	Accordingly, we analyze the summation component in Eq.~\eqref{hat_A_chi} as $S_N = \sum_{n=1}^N f^{\chi \chi}(n)$, where the kernel function is defined as $f^{\chi \chi}(x) = ( Y_{\chi}^2 + x^2 L_{\text{pa}}^2 )^{-\frac{1}{2}}$. Since $f^{\chi \chi}(x)$ is a positive, monotonically decreasing function for $x > 0$, we invoke the integral inequalities to bound the summation by
	\begin{equation}\label{inequalities}
		\int_{1}^{N+1} f^{\chi \chi}(x) \mathrm{d} x \le \sum_{n=1}^N f^{\chi \chi}(n) \le \int_{0}^{N} f^{\chi \chi}(x) \mathrm{d} x.
	\end{equation}
	Utilizing the standard integration result: $\int (Y^2+x^2)^{-1/2} \mathrm{d}x = \ln( x + \sqrt{Y^2+x^2} ) + C$, the lower and upper bounds in \eqref{inequalities} are derived as
	\begin{equation}
	\small
	\begin{aligned}
		\mathcal{L}_{\text{C}}
		&=\frac{1}{L_{\text{pa}}}
		\ln\frac{(N+1)L_{\text{pa}}+\sqrt{Y_{\chi}^2+(N+1)^2L_{\text{pa}}^2}}
		{L_{\text{pa}}+\sqrt{Y_{\chi}^2+L_{\text{pa}}^2}},\\
		\mathcal{U}_{\text{C}}
		&=\frac{1}{L_{\text{pa}}}
		\ln\frac{NL_{\text{pa}}+\sqrt{Y_{\chi}^2+N^2L_{\text{pa}}^2}}
		{Y_{\chi}} .
	\end{aligned}
	\end{equation}
	It can be examined that both bounds converge to the same scaling order: $\mathcal{O}(\ln N)$ as $N \! \to \! \infty$. By the Squeeze Theorem, the summation $S_N$ scales as $\mathcal{O}(\ln N)$. Substituting this result back into \eqref{hat_A_chi} and the above dominant-gain expressions, we obtain the scaling law of the array gain and multiplexing gain of the C-PASS as:
	\begin{equation}
		\bar{G}_{\text{C}}^{\text{A}} \sim \mathcal{O}\left(\frac{\ln^2 N}{N}\right), \quad \bar{G}_{\text{C}}^{\text{M}} \sim \mathcal{O}\left(P_T \frac{\ln^4 N}{N^2}\right).
	\end{equation}
	
	A similar derivation applies to the end-fed PASS. With the same PA fine-tuning strategy, the first $N$ PAs are phase-aligned at the FU and the last $N$ PAs at the BU, yielding the same array-gain order $\mathcal{O}(\ln^2 N/N)$. However, due to unidirectional propagation, $\det(\mathbf{H}_{\text{E}}^{\text{eff}})=0$, so the multiplexing gain is zero. This is consistent with the scaling result in~\cite{ouyang2025array}.	
\end{proof}

\begin{table}[t]
	\centering
	\resizebox{\columnwidth}{!}{
	\begin{tabular}{|l|c|c|c|} 
		\hline 
		\textbf{Architecture} & \textbf{DoF} & \textbf{Array Gain} & \textbf{Multiplexing Gain} \\
		\hline 
		Center-Fed PASS & 2 & $\mathcal{O}\left( \frac{\ln^2 N}{N} \right)$ & $\mathcal{O}\left( P_T \frac{\ln^4 N}{N^2} \right)$ \\
		\hline 
		End-Fed PASS & 1 & $\mathcal{O}\left( \frac{\ln^2 N}{N} \right)$ & 0 \\
		\hline 
	\end{tabular}
}
	\caption{Comparison between center-fed and end-fed PASS.}
	\label{tab:analysis}
	\vspace{-5mm}
\end{table}

Table~\ref{tab:analysis} summarizes Theorems~1 and~2, showing that C-PASS preserves the array-gain order of the end-fed PASS while additionally providing doubled DoF and a non-zero multiplexing-gain term.

\begin{remark}
\normalfont Practical T-junction non-idealities, such as insertion loss, finite tuning resolution, and residual phase mismatch, can be modeled as perturbations of the branch feeding coefficients. These perturbations mainly affect the finite-SNR array and multiplexing gain magnitudes. Since the nominal C-PASS channel is full-rank under non-degenerate geometry and feeding configurations, bounded hardware perturbations do not remove the structural DoF advantage unless they induce exact determinant cancellation. Practical C-PASS design should therefore avoid near-zero determinant operating points. A detailed hardware-level analysis and robust parameter optimization under T-junction impairments are left for future work.
\end{remark}

\section{Numerical Results}
In this section, we provide numerical results to validate the analytical derivations for the DoF and power scaling laws of the C-PASS and the conventional end-fed PASS. They also demonstrate the significant capacity enhancement of the C-PASS over the conventional architecture, confirming the effectiveness of the proposed C-PASS. Unless otherwise specified, the simulations operate at a carrier frequency of $f_c = 28$ GHz with a waveguide refractive index of $n_{\text{eff}}=1.4$. The two users are located at vertical distances of $Y_{\text{F}}=35$ m and $Y_{\text{B}}=40$ m. We employ the PA spacing of $L_{\text{pa}}=1$ m and the input port separation of $L_{\text{in}} = 1.25\lambda_g$. To ensure hardware feasibility, position tuning is strictly constrained to $|\Delta_n| \le 0.01$ m. The transmit power is set to $P_T = 30$ dBm, and the noise power of $N_0 = -80$ dBm.

\begin{figure*}[t]
	\centering
	\subfloat[DoF comparison.]{
		\includegraphics[width=0.32\linewidth]{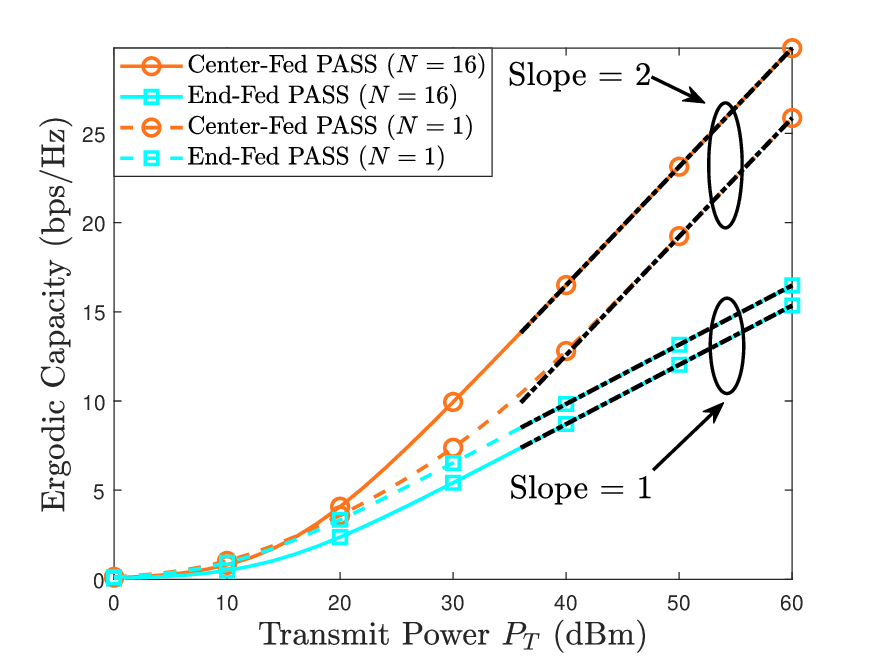}
		\label{fig:DoF}
	}
	\subfloat[Effective gain of C-PASS.]{
		\includegraphics[width=0.32\linewidth]{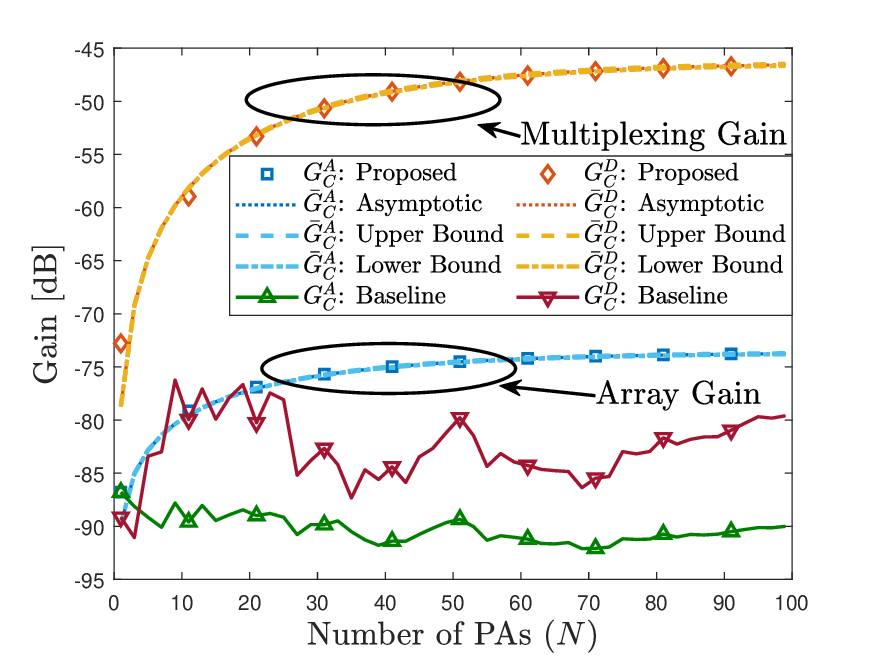}
		\label{fig:center_gain}
	}
	\subfloat[Capacity enhancement of C-PASS.]{
		\includegraphics[width=0.32\linewidth]{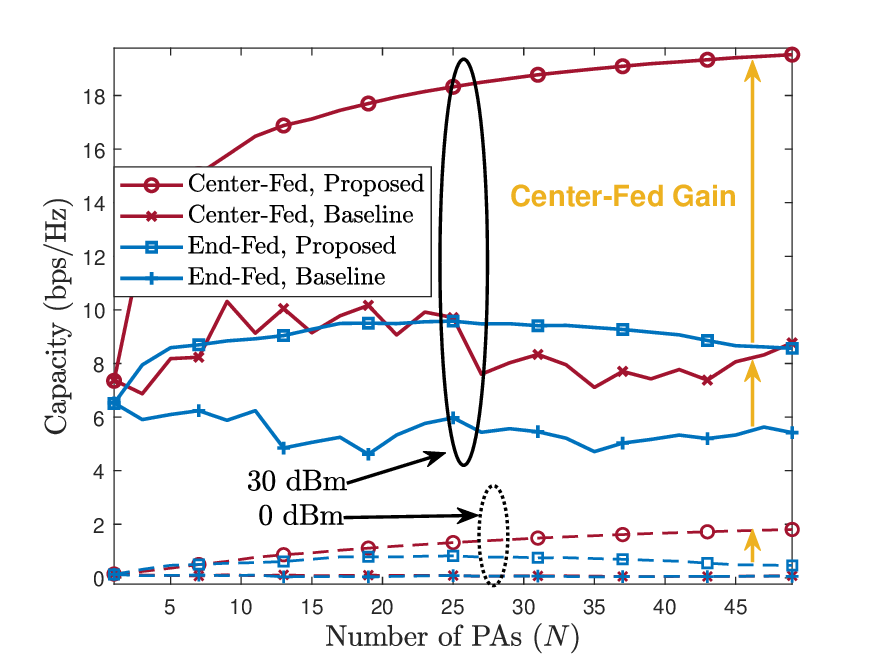}
		\label{fig:capacity_center_end}
	}
	\caption{Numerical results of the proposed C-PASS.}
	\label{fig:simulation_all}
	\vspace{-2mm}
\end{figure*}

Fig.~\ref{fig:DoF} plots the ergodic capacity versus transmit power to illustrate the DoF. The dashed lines are reference curves with slopes of 1 and 2. In the high-SNR regime, the end-fed and center-fed architectures closely follow the slope-1 and slope-2 references, respectively, validating \emph{Theorem 1}. The capacity advantage of C-PASS becomes more pronounced at high $P_T$. For instance, at $P_T = 60$ dBm, C-PASS achieves nearly \emph{double} the capacity of its end-fed counterpart. As $N$ increases from $1$ to $16$, C-PASS continues to improve, whereas the end-fed PASS degrades because its serial extension increases the average PA-to-user distance. The center-fed topology mitigates this path-loss increase by distributing the PA groups around the input ports.

Fig.~\ref{fig:center_gain} plots the array and multiplexing gains of C-PASS versus the number of PAs, comparing the proposed position-tuning scheme with the uniform deployment baseline. Here, ``Proposed'' and ``Baseline'' denote fine-tuned PA deployment and uniform PA deployment without position tuning, respectively. Under uniform deployment, both $G_{\text{C}}^{\text{A}}$ and $G_{\text{C}}^{\text{M}}$ fluctuate irregularly because the deterministic PA-dependent LoS phases are not aligned as $N$ varies. We therefore focus on the proposed position-tuning scheme to validate \emph{Theorem 2}. The simulated gains closely match the asymptotic analyses and remain within the theoretical upper and lower bounds. For $N>5$, this agreement validates the scaling orders $\mathcal{O}(\ln^2 N/N)$ and $\mathcal{O}(P_T\ln^4 N/N^2)$.

Fig.~\ref{fig:capacity_center_end} plots the capacity versus the number of PAs for center-fed and end-fed PASS under two PA deployment schemes at $P_T=0$ dBm and $30$ dBm. Here, ``center-fed gain'' denotes the C-PASS advantage over the end-fed PASS due to the combined array and multiplexing gains. C-PASS consistently outperforms the end-fed PASS across the considered deployment schemes, owing to the additional multiplexing gain enabled by the center-fed architecture. As shown in Table~\ref{tab:analysis}, under PA fine-tuning, C-PASS achieves a receive-strength improvement of order $\mathcal{O}(P_T\ln^4 N/N^2)$. This trend agrees with Fig.~\ref{fig:capacity_center_end}, where the C-PASS capacity increases with both $P_T$ and $N$. For instance, at $P_T=30$ dBm and $N=50$, C-PASS provides a $3.59$ dB capacity improvement over the end-fed architecture.

\section{Conclusion}
In this letter, a novel architecture of C-PASS has been proposed. A basic signal model characterizing the bidirectional in-waveguide propagation was presented. To evaluate the performance of the C-PASS, closed-form expressions for the DoF and power scaling laws were derived and compared with those of the conventional end-fed PASS. Numerical results validated the analytical derivations, confirming the effectiveness of C-PASS in significantly enhancing communication performance. These results motivate future research on C-PASS-enabled wireless networks, which are envisioned to deliver superior performance by exploiting the doubled DoF and additional multiplexing gain.

\bibliographystyle{IEEEtran}
\bibliography{reference/mybib}

\begin{thebibliography}{10}
\providecommand{\url}[1]{#1}
\csname url@samestyle\endcsname
\providecommand{\newblock}{\relax}
\providecommand{\bibinfo}[2]{#2}
\providecommand{\BIBentrySTDinterwordspacing}{\spaceskip=0pt\relax}
\providecommand{\BIBentryALTinterwordstretchfactor}{4}
\providecommand{\BIBentryALTinterwordspacing}{\spaceskip=\fontdimen2\font plus
\BIBentryALTinterwordstretchfactor\fontdimen3\font minus
  \fontdimen4\font\relax}
\providecommand{\BIBforeignlanguage}[2]{{%
\expandafter\ifx\csname l@#1\endcsname\relax
\typeout{** WARNING: IEEEtran.bst: No hyphenation pattern has been}%
\typeout{** loaded for the language `#1'. Using the pattern for}%
\typeout{** the default language instead.}%
\else
\language=\csname l@#1\endcsname
\fi
#2}}
\providecommand{\BIBdecl}{\relax}
\BIBdecl

\bibitem{liu2025pinching}
Y.~Liu, Z.~Wang, X.~Mu, C.~Ouyang, X.~Xu, and Z.~Ding, ``Pinching-antenna
  systems: Architecture designs, opportunities, and outlook,'' \emph{{IEEE}
  Commun. Mag.}, vol.~64, no.~1, pp. 190--196, 2026.

\bibitem{liu2025pinching2}
Y.~Liu, H.~Jiang, X.~Xu, Z.~Wang, J.~Guo, C.~Ouyang, X.~Mu, Z.~Ding,
  A.~Nallanathan, G.~K. Karagiannidis, and R.~Schober, ``Pinching-antenna
  systems {(PASS)}: A tutorial,'' \emph{{IEEE} Trans. Commun.}, vol.~74, pp.
  4881--4918, 2026.

\bibitem{yang2025pinching}
Z.~Yang, N.~Wang, Y.~Sun, Z.~Ding, R.~Schober, G.~K. Karagiannidis, V.~W. Wong,
  and O.~A. Dobre, ``Pinching antennas: Principles, applications and
  challenges,'' \emph{{IEEE} Wireless Commun.}, vol.~33, no.~2, pp. 175--184,
  2026.

\bibitem{wang_model}
Z.~Wang, C.~Ouyang, X.~Mu, Y.~Liu, and Z.~Ding, ``Modeling and beamforming
  optimization for pinching-antenna systems,'' \emph{{IEEE} Trans. Commun.},
  vol.~73, no.~12, pp. 13\,904--13\,919, 2025.

\bibitem{shan2025secure}
S.~Shan, C.~Ouyang, Y.~Li, and Y.~Liu, ``Secure multicast communications with
  pinching-antenna systems {(PASS)},'' \emph{arXiv preprint arXiv:2509.16045},
  2025.

\bibitem{gan2025dual}
X.~Gan, Z.~Wang, and Y.~Liu, ``Dual-scale antenna deployment for pinching
  antenna systems,'' \emph{{IEEE} Trans. Commun.}, early access, 2026. doi:
  10.1109/TCOMM.2026.3698896.

\bibitem{lv2025pinching}
S.~Lv, M.~Li, Q.~Li, and Y.~Liu, ``Pinching-antenna systems {(PASS)}-enabled
  {UAV} delivery,'' \emph{{IEEE} Trans. Commun.}, early access, 2026. doi:
  10.1109/TCOMM.2026.3681648.

\bibitem{ouyang2025array}
C.~Ouyang, Z.~Wang, Y.~Liu, and Z.~Ding, ``Array gain for pinching-antenna
  systems {(PASS)},'' \emph{{IEEE} Commun. Lett.}, vol.~29, no.~6, pp.
  1471--1475, 2025.

\bibitem{hong2018high}
L.~Hong, S.~Xiao, X.~Deng, R.~Pu, and L.~Shen, ``High-efficiency tunable
  {T}-shaped beam splitter based on one-way waveguide,'' \emph{Journal of
  Optics}, vol.~20, no.~12, p. 125002, 2018.

\bibitem{reichel2016broadband}
K.~S. Reichel, R.~Mendis, and D.~M. Mittleman, ``A broadband terahertz
  waveguide {T}-junction variable power splitter,'' \emph{Sci. Rep.}, vol.~6,
  no.~1, p. 28925, 2016.

\bibitem{li2015broadband}
T.~Li and W.~Dou, ``Broadband substrate-integrated waveguide {T}-junction with
  arbitrary power-dividing ratio,'' \emph{Elect. Lett.}, vol.~51, no.~3, pp.
  259--260, 2015.

\bibitem{pozar2011microwave}
D.~M. Pozar, \emph{Microwave engineering}.\hskip 1em plus 0.5em minus
  0.4em\relax John wiley \& sons, 2011.

\end{thebibliography}

\end{document}